\DeclareMathAlphabet{\altmathcal}{OMS}{cmsy}{m}{n}

\documentclass[preprint, journal]{vgtc}
\onlineid{0}
\preprinttext{}

\title{Petal-X: Human-Centered Visual Explanations to Improve Cardiovascular Risk Communication}

\author{%
  \authororcid{Diego Rojo}{0000-0003-2672-0738},
  \authororcid{Houda Lamqaddam}{0000-0001-6738-1934}, 
  \authororcid{Lucija Gosak}{0000-0002-8742-6594}, and 
  \authororcid{Katrien Verbert}{0000-0001-6699-7710} 
}

\authorfooter{
  \item Diego Rojo is with KU Leuven.
        E-mail: diego.rojogarcia@kuleuven.be
  \item Houda Lamqaddam is with University of Amsterdam.
  	E-mail: h.lamqaddam@uva.nl
  \item Lucija Gosak is with University of Maribor.
  	E-mail: lucija.gosak2@um.si
 \item Katrien Verbert is with KU Leuven.
    E-mail: katrien.verbert@kuleuven.be
}

\abstract{%
  Cardiovascular diseases (CVDs), the leading cause of death worldwide, can be prevented in most cases through behavioral interventions. Therefore, effective communication of CVD risk and projected risk reduction by risk factor modification plays a crucial role in reducing CVD risk at the individual level. However, despite interest in refining risk estimation with improved prediction models such as SCORE2, the guidelines for presenting these risk estimations in clinical practice remained essentially unchanged in the last few years, with graphical score charts (GSCs) continuing to be one of the prevalent systems. This work describes the design and implementation of Petal-X, a novel tool to support clinician-patient shared decision-making by explaining the CVD risk contributions of different factors and facilitating what-if analysis. Petal-X relies on a novel visualization, Petal Product Plots, and a tailor-made global surrogate model of SCORE2, whose fidelity is comparable to that of the GSCs used in clinical practice. We evaluated Petal-X compared to GSCs in a controlled experiment with 88 healthcare students, all but one with experience with chronic patients. The results show that Petal-X outperforms GSC in critical tasks, such as comparing the contribution to the patient's 10-year CVD risk of each modifiable risk factor, without a significant loss of perceived transparency, trust, or intent to use. Our study provides an innovative approach to the visualization and explanation of risk in clinical practice that, due to its model-agnostic nature, could continue to support next-generation artificial intelligence risk assessment models.
}

\teaser{
    \centering
    \includegraphics[width=0.42\textwidth]{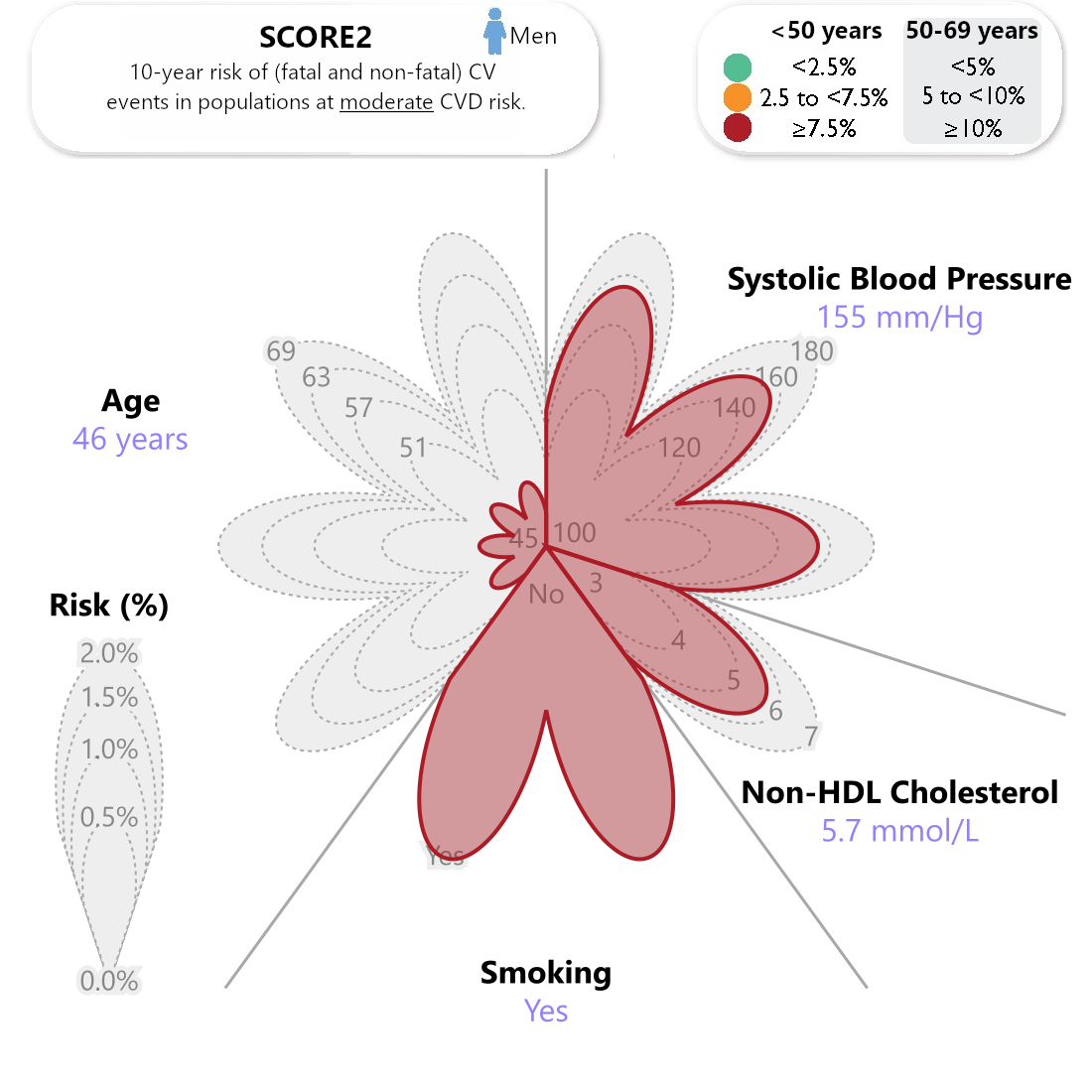}
    \includegraphics[width=0.42\textwidth]{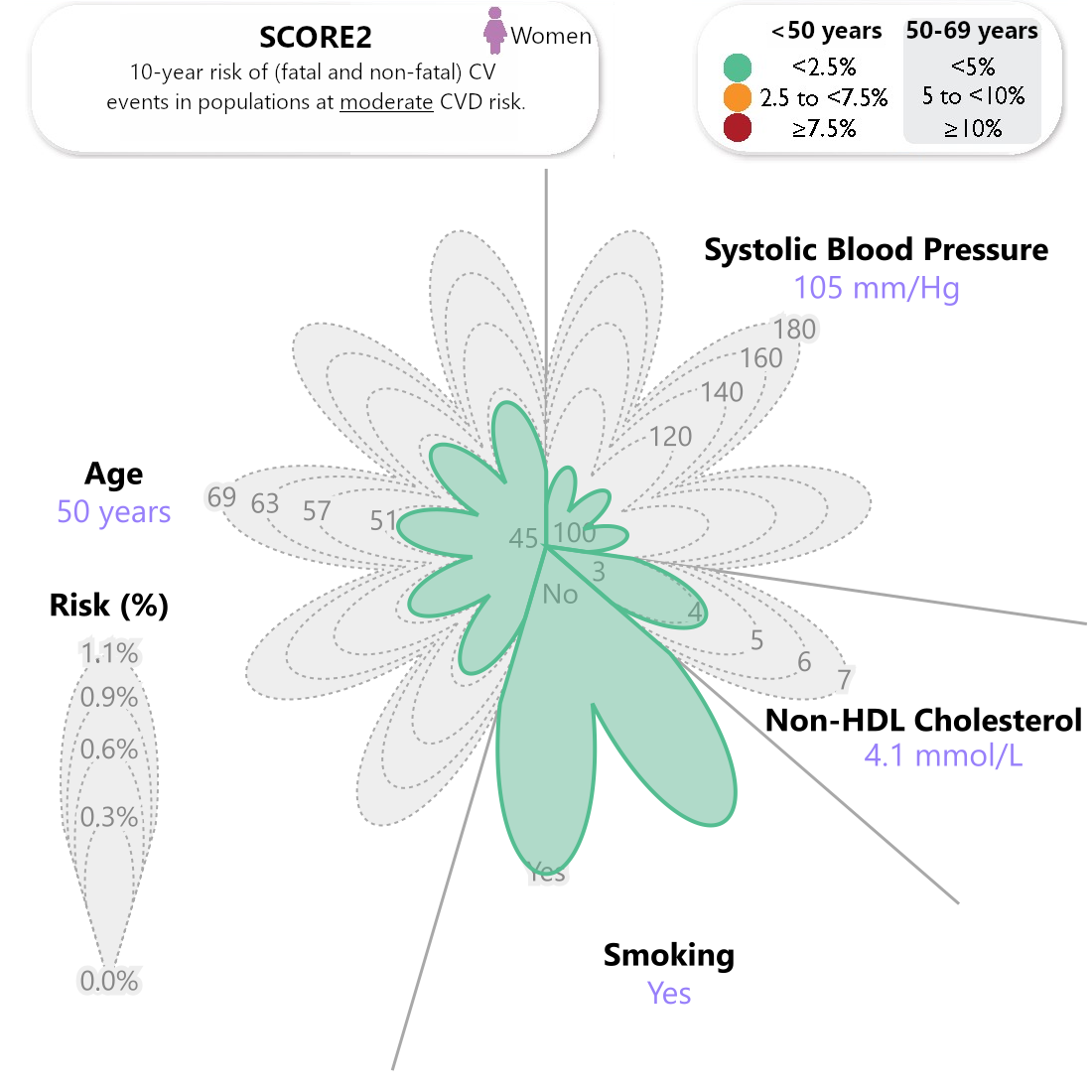}
    \caption{Petal-X human-centered 10-year CVD risk visual explanations of a male and female patient.}
    \label{fig:teaser}
}

\renewcommand{\manuscriptnotetxt}{}

\nocopyrightspace

\graphicspath{{figs/}{figures/}{pictures/}{images/}{./}} 

\usepackage{tabu}                      
\usepackage{booktabs}                  
\usepackage{lipsum}                    
\usepackage{mwe}                       

\usepackage{mathptmx}                  

\usepackage[dvipsnames,svgnames]{xcolor}
\usepackage{caption}
\usepackage{subcaption}

\usepackage{enumitem}
\setlist{topsep=0pt}

\usepackage{amsthm}
\usepackage{amsmath}
\usepackage{amssymb}
\usepackage{bm}
\newtheorem{proposition}{Proposition}
\newtheorem{corollary}{Corollary}

\usepackage{tikz}
\definecolor{redrisk}{HTML}{ad1e28}
\definecolor{orangerisk}{HTML}{f79228}
\definecolor{greenrisk}{HTML}{55bd92}

\begin{document}


\firstsection{Introduction}

\maketitle

Clinical risk communication is the process by which healthcare professionals engage with patients to discuss the risks and benefits of treatments or behavioral (risk-reducing) changes~\cite{edwards1999}. As the demand for shared clinical decision-making continues to grow, tools that effectively facilitate risk communication by healthcare professionals are becoming increasingly important~\cite{zipkin2014, alston2012}.

Central to such risk communication tools are clinical prediction models, which have progressed from simple scoring systems based on the presence or absence of risk factors to more sophisticated machine learning (ML) algorithms~\cite{grobman2006}. However, achieving a balance between predictive performance and model interpretability is crucial for these ML models to be useful in risk communication. Specifically, interpretability is necessary for domain experts to trust ML models and confidently incorporate them into their practices~\cite{zhou2018}.

The field of explainable artificial intelligence (XAI) aims to develop ML systems that have high predictive performance yet are comprehensible to a given audience. The two main approaches are (1) using inherently interpretable models, i.e., models that are transparent by design, and (2) applying post-hoc techniques to explain an existing ML model~\cite{barredo2020}. Although risk prediction models commonly used in clinical practice are developed using logistic regression or Cox regression and thus inherently interpretable~\cite{Moons2015}, even these transparent models may require post-hoc visualizations for audiences with little ML expertise, such as healthcare professionals or their patients, to understand them~\cite{barredo2020}.

However, a recent review of visual analytics tools for XAI in healthcare highlighted the need for visual explanations of clinical prediction models aimed at lay audiences~\cite{Ooge2022}. In this paper, we propose and evaluate Petal-X, a tool tailored to lay audiences that provides post-hoc explanations of SCORE2 (the European state-of-the-art cardiovascular disease risk prediction model~\cite{SCORE22021}) using a novel visual representation, Petal Product Plots.

The paper is structured as follows. In Section \ref{section:predviz}, we characterize the cardiovascular risk domain. This is followed by Sections \ref{section:surrogate} and \ref{section:ppp}, which describe the design process of the two main building blocks of Petal-X: a tailor-made SCORE2 global surrogate model and Petal Product Plots. Section \ref{section:petalx} shows how they are both integrated into Petal-X visual explanations. Finally, in Section \ref{section:evaluation}, we describe the Petal-X evaluation, the results of which are discussed in Section \ref{section:discussion}.


\section{Predicting and visualizing cardiovascular risk} \label{section:predviz}
Cardiovascular diseases (CVDs), such as stroke and coronary heart disease, are the most common non-communicable diseases worldwide, accounting for approximately 18.6 million deaths in 2019~\cite{SCORE22021}. Guidelines for CVD prevention advocate the use of CVD risk predictions to identify those patients at the highest risk of CVD and to support shared decision-making by the patient and their healthcare professional \cite{cvdguidelines2016, cvdguidelines2021}. Multiple CVD risk prediction models have been developed to achieve this goal, focusing on different populations or patient groups (e.g., apparently healthy people or patients with type 2 diabetes mellitus)\cite{Rosello2020, cvdguidelines2021}. Furthermore, some of these risk prediction models have also been made available in clinical practice through various tools that simplify the estimation and communication of cardiovascular prognoses.

In the remainder of this section, we first introduce SCORE2, the model that we use to estimate CVD risk, then analyze the visualizations used by existing tools to communicate CVD risk in clinical settings, and finally characterize the challenges that our proposed tool should support.

\subsection{SCORE2 CVD risk model} 
Systemic Coronary Risk Estimation 2 (SCORE2) is a prediction model to estimate the risk of 10-year fatal and non-fatal cardiovascular disease in individuals without previous CVD or diabetes aged 40 to 69~\cite{SCORE22021}. SCORE2 has sex-specific model coefficients and has been calibrated for four distinct European regions characterized by varying CVD risk levels (low, moderate, high risk and very high risk).

We opt for this model, as it is the standard approach suggested to communicate the risk of CVD in apparently healthy people by the 2021 European Society of Cardiology (ESC) guidelines on cardiovascular disease prevention in clinical practice~\cite{cvdguidelines2021}, thus adapting to the European context in which the evaluation of this research was carried out. Moreover, SCORE2 is a nonproprietary model whose implementation details have been published by the SCORE2 Working Group~\cite{SCORE22021}, allowing us not only to query it, but also to know the predictors it uses and to understand how they are mapped to a risk estimate. 

In particular, the SCORE2 model uses five predictors: age, smoking status (yes/no), systolic blood pressure, total cholesterol, and high-density lipoprotein (HDL) cholesterol. These predictors are mapped to a 10-year CVD risk percentage using the following formula:
\begin{equation}\label{eq:score2}
1-\exp\left(-\exp \left(\text{S}_{1}+\text{S}_{2} \times \ln \left(-\ln \left(\lambda^{\exp \left(\sum \beta\left(x-x_{\text{cen}}\right)\right)}\right)\right)\right)\right)
\end{equation}
where $\text{S}_{1}$, $\text{S}_{2}$ are region- and sex-specific scaling factors, $\lambda$ is the sex-specific baseline survival, $\beta$ are sex-specific model coefficients, $x$ are the values of the patient's risk factors together with their age-interactions, and $x_{\text{cen}}$ are the values at which each risk factor was centered. The concrete values of the parameters and coefficients for each sex and risk region can be found in the supplementary material provided by the SCORE2 Working Group~\cite{SCORE22021}.

\subsection{CVD risk visualization in clinical practice} 
In addition to a reliable patient estimate of CVD risk based on models such as SCORE2, the reduction of individual CVD risk requires effective communication of both individual risk and expected reduction of risk through risk factor treatments~\cite{cvdguidelines2021}. Although risk communication is complex and there is no one-size-fits-all communication strategy, visual representation of risk can improve patient understanding and satisfaction~\cite{cvdguidelines2021, zipkin2014, spiegelhalter2017}. 

In clinical settings, two graphical representations are predominant for presenting risk prediction models: graphical score charts (GSCs) and nomograms~\cite{bonett2019}. In addition to facilitating communication, both tools allow one to calculate the output of the risk prediction model for new individuals or data. Nomograms are graphical representations of the mathematical formula of the original model. In contrast, graphical score charts are simplified representations that just approximate the underlying model but are easier to understand and use~\cite{Moons2015, bonett2019}. Consequently, Bonnet~et~al.~\cite{bonett2019} recommend graphical score charts for healthcare professionals and patients but consider that nomograms are only suitable for healthcare professionals.

Increasingly, websites and mobile apps are also being used to calculate and communicate risk predictions, generally through interactive visualization. Similarly to nomograms, these digital tools rely on the more complex original model to calculate the risk predictions. Despite this, they can be aimed at healthcare professionals and the general public as GSCs since the model is generally hidden from the end user~\cite{bonett2019}. However, by hiding the model and solely providing the risk outcome, the model functioning becomes unclear for the end user, i.e., the model becomes a black box. 

In the following, we analyze the concrete visualizations used in the main tools suggested to communicate CVD risk in the 2021 ESC CVD guidelines~\cite{cvdguidelines2021}: the SCORE2 graphical score charts provided by the ESC, \href{https://www.heartscore.org}{heartscore.org} and \href{https://u-prevent.com}{u-prevent.com} websites, and the ESC CVD Risk Calculation mobile application. 

\subsubsection{SCORE2 graphical score charts}
\begin{figure}[h!t]
\centering
\includegraphics[width=0.9\linewidth]{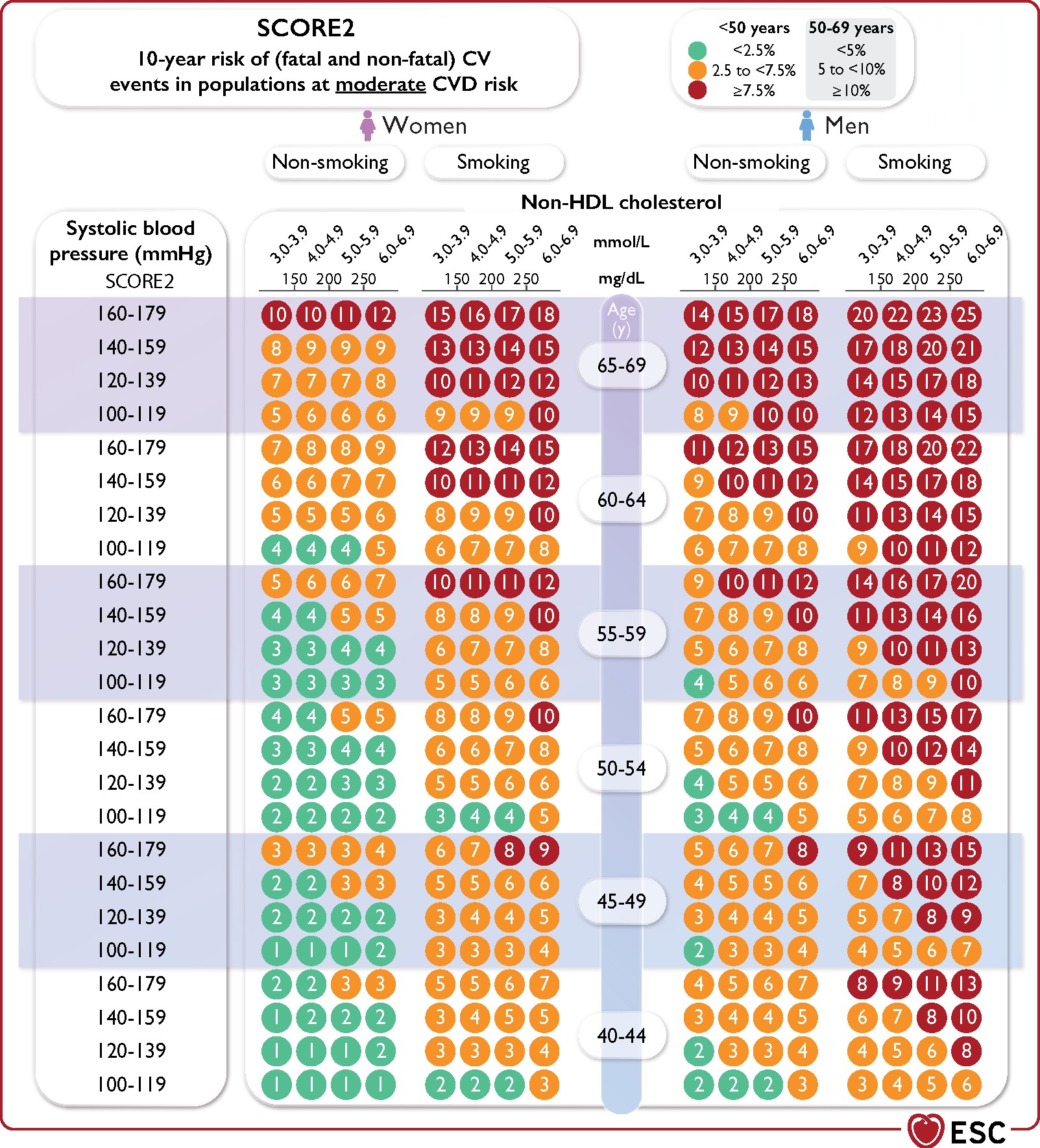}
\caption{SCORE2 graphical score chart for populations at moderate CVD risk. Adapted from~\cite{cvdguidelines2021}, modified.}
\label{fig:gsc}
\end{figure}

Graphical score charts (GSCs) are derived by cross-tabulating the model predictors, with each of them recoded into several categories. The estimated risk for a person with values of each predictor at the midpoint of the predictor's range is encoded in each cell using color. The numeric value of the estimated risk is sometimes also included in the cell. Graphical score charts usually require some simplification of the clinical prediction model as they can only accommodate a limited number of predictors and require continuous predictors to be presented as categories.

In the case of the SCORE2 GSCs (see \autoref{fig:gsc}), the number of predictors of the original SCORE2 model is reduced to four by combining total cholesterol and HDL-cholesterol into one unique predictor, non-HDL cholesterol. Furthermore, the continuous predictors are categorized into $5$ years intervals (age), $20\;mmHG$ intervals (systolic blood pressure), and $1\;mmol\,/\,L$ intervals (non-HDL cholesterol). Finally, risk estimates are encoded in each cell using the (rounded) numerical value and an age-dependent color scale.

\subsubsection{Websites and mobile apps visualizations} 
The analyzed websites and mobile applications visualize the patient's SCORE2 10-year CVD risk once the user inputs the patient's clinical data using various types of visuals, such as gauge charts and bar charts. At the time of writing this, \href{https://www.heartscore.org}{heartscore.org} is still updating its system to the 2021 ESC CVD guidelines~\cite{cvdguidelines2021}, so screenshots from \autoref{fig:examples}a,b correspond to the SCORE-based version.

\begin{figure}[h!t]
     \centering
     \begin{subfigure}[t]{0.37\linewidth}
         \centering
         \includegraphics[width=\textwidth]{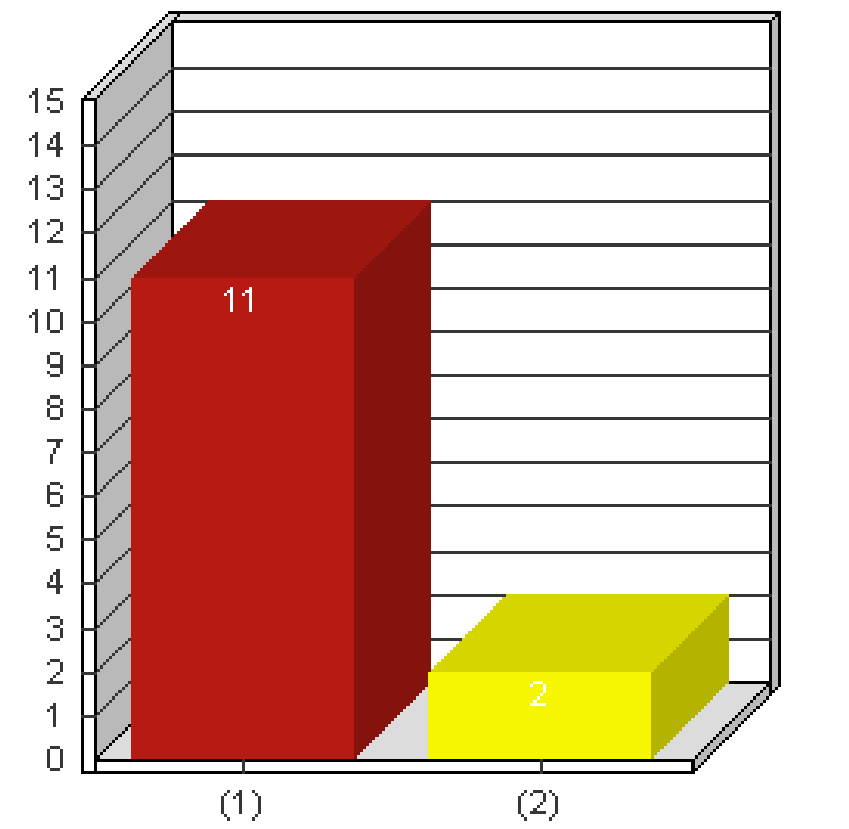}
         \caption{Bar chart}
         \label{fig:heartbar}
     \end{subfigure}
     \hfill
     \begin{subfigure}[t]{0.5\linewidth}
         \centering
         \includegraphics[width=\textwidth]{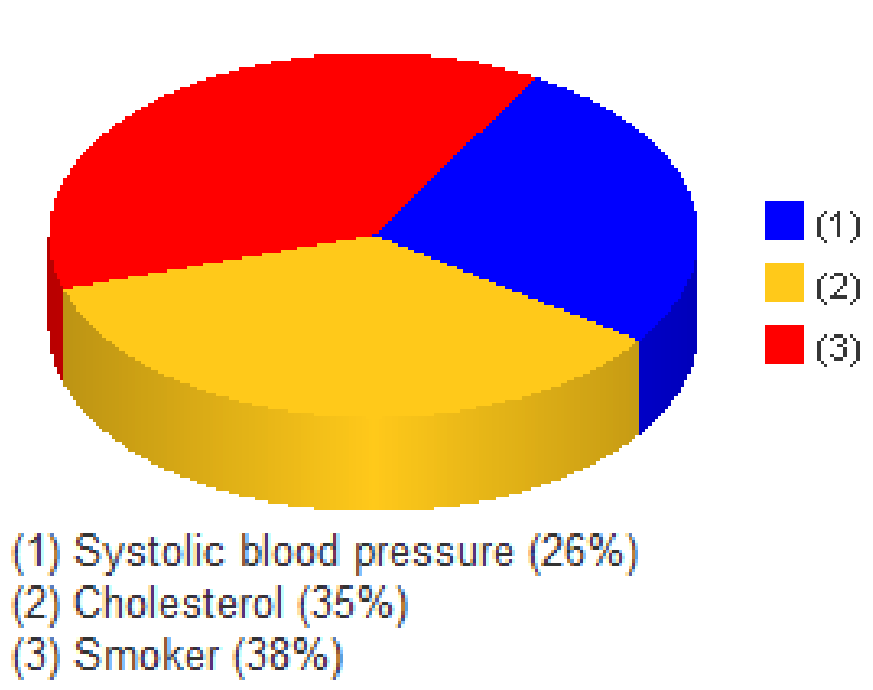}
         \caption{Factor contribution pie chart}
         \label{fig:heartpie}
     \end{subfigure}
     \begin{subfigure}[t]{0.45\linewidth}
         \centering
         \includegraphics[width=\textwidth]{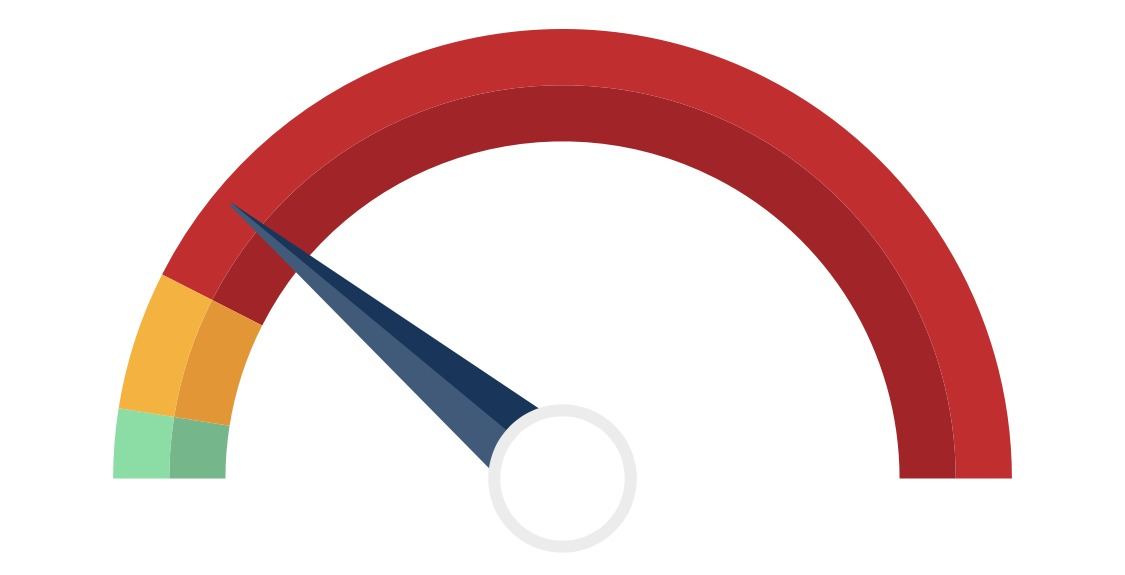}
         \caption{Gauge Chart}
         \label{fig:appgauge}
     \end{subfigure}
     \hfill
     \begin{subfigure}[t]{0.5\linewidth}
         \centering
         \includegraphics[width=\textwidth]{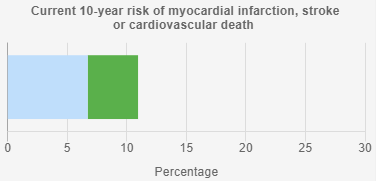}
         \caption{Bar chart with risk reduction}
         \label{fig:ubar}
     \end{subfigure}
        \caption{CVD risk visualizations screenshots from \href{https://www.heartscore.org}{heartscore.org} (a, b), ESC CVD Risk Calculation app (c) and \href{https://u-prevent.com}{u-prevent.com}~(d).}
        \label{fig:examples}
\end{figure}

Unlike SCORE2 GSCs, these tools do not provide a global explanation of the model. However, both websites allow users to visualize the projected risk reduction in the patient's 10-year CVD risk if specific intervention goals are achieved (see \autoref{fig:examples}a,d). In addition,  \href{https://www.heartscore.org}{heartscore.org} also provides a 3D pie chart that shows the relative contribution of each modifiable risk factor to the total 10-year CVD risk (see \autoref{fig:heartpie}).

\subsection{Domain problem characterization} \label{sec:challenges}
The analysis of current CVD risk communication tools and discussions with healthcare professionals allowed us to derive the following challenges that Petal-X should support for efficient clinician-patient CVD risk communication.

\begin{itemize}[noitemsep]
    \item[\textbf{C0}.]\textbf{Identify CVD risk.} Users should be able to identify the patient's 10-year CVD risk. 
    
    \item[\textbf{C1}.]\textbf{Understand the contribution of each risk factor.} Users should be able to understand the contribution of each risk factor to the total 10-year CVD risk. 

    \item[\textbf{C2}.]\textbf{Perform \textit{what-if} analysis.} Users should be able to comprehend \textit{what} the model output risk reduction would be \textit{if} the risk factor values of the patient were changed.
\end{itemize}

\section{SCORE2 surrogate model} \label{section:surrogate}
To allow patients or healthcare professionals to address the three challenges described in Section~\ref{sec:challenges}, the 10-year CVD risk percentage output given by the SCORE2 model is insufficient; more details of the model are needed to make its functioning clear and comprehensible to the target audience.

Despite SCORE2 being considered an (inherently) interpretable model under classifications such as that of Doran~et~al.~\cite{Doran2018} due to the availability of the mathematical mapping shown in \autoref{eq:score2}, it is clear that this formula does not make the functioning of SCORE2 clearer or more comprehensible to a lay audience. The main obstacle to the comprehensibility of SCORE2 is the number of non-linear functions applied to the otherwise simple linear combination of centered risk factors, i.e. $\sum \beta\left(x-x_{\text{cen}}\right)$ in \autoref{eq:score2}. In particular, of the three levels of transparency described by Lipton~\cite{lipton2018}, this has a major impact on the simulability of the SCORE2 model, as these non-linearities greatly increase the time, human cognition, and mathematical background required to produce a prediction from the input values and model parameters.

Therefore, post-hoc explanations are needed to describe the functioning of the SCORE2 model to patients. In fact, although generally not described as a post-hoc explanation method in the healthcare domain \cite{bonett2019}, current communication tools such as SCORE2 graphical score charts are indeed a global surrogate model, i.e., a more comprehensible model that approximates the more complex model.

In the remainder of this section, we discuss the design and development of our post-hoc explanation of the SCORE2 model, describe the dataset that was used to train it, and validate the developed surrogate models by comparing their fidelity to that of the SCORE2 graphical score charts.

\subsection{Global surrogate model}
As we have already noted that the SCORE2 model is not suitable to be explained directly to the target audience, the main decision left to navigate the XAI taxonomy proposed by Vijay~et~al.~\cite{vijay2021} is whether to use global or local post-hoc explanations to understand the SCORE2 model. We opt for global surrogate models, as the same validation process that the SCORE2 model underwent could be applied directly to validate them, a requirement before any auxiliary method for explaining the SCORE2 model is widely adopted in clinical practice~\cite{bonett2019}. 

We use a linear regression as the surrogate model, since its direct interpretation is closely aligned with challenges C1 and C2. In particular, for each sex-specific SCORE2 model, our linear regression surrogate takes the form
\begin{equation}\label{eq:surrogate}
\hspace*{-5pt}y = \alpha_{age} \cdot z_{age} + \alpha_{sbp} \cdot z_{sbp} +  \alpha_{smoking} \cdot z_{smoking} +  \alpha_{nonhdl} \cdot z_{nonhdl}
\end{equation}
where $y$ corresponds to the 10-year CVD risk predictions of the SCORE2 model, $z_{var}$ are the min-max normalized independent variables and $\alpha_{var}$ its corresponding coefficients. 

To increase the interpretability of the linear surrogate model, three key design decisions were taken:
\begin{enumerate}
    \item \textbf{Non-HDL cholesterol.} Likewise to SCORE2 graphical score charts, we combined total cholesterol and HDL-cholesterol into one unique input variable: non-HDL cholesterol (\textit{nonhdl} in \autoref{eq:surrogate}). This reduces the number of features in the explanation, which increases its simulatability, and evens out the direction of all predictors' effects (HDL cholesterol coefficient has reversed sign in the original SCORE2 model). Furthermore, non-HDL cholesterol is considered a reasonable alternative treatment goal for all patients \cite{cvdguidelines2021}.
    
    \item \textbf{Min-max normalization.} Continuous input variables (age, systolic blood pressure, and non-HDL cholesterol) have been normalized using the min-max method, i.e., 
     \[
        z = \frac {x-\min}{\max-\min} 
     \]
    where, for each variable, $\min$ and $\max$ are the corresponding values from \autoref{tab:rangetable}. The use of this normalization allows for an easier interpretation of the linear regression coefficient. In particular, the linear regression coefficient $\alpha_{var}$ can be interpreted as the estimated risk increase when going from the most optimal value considered by our model (the minimum or `no' in the case of $smoking$) to the worst value considered by our model (the maximum or `yes' in the case of $smoking$). 
    
    \item \textbf{No-intercept linear regression.} In our effort to increase the interpretability for our lay audience, we considered using a no-intercept linear regression, which removes the need to interpret an additional parameter. Using no-intercept linear regressions is generally justified based on the existence of a priori reasons to believe the intercept is zero~\cite{eisenhauer2003}. This is not our case, as the moderate risk region SCORE2 10-year CVD risk of a patient with a zero value for each normalized input variable $z_{var}$, i.e., a non-smoking individual of 45 years old with $100\;mmHg$ systolic blood pressure and $3\;mmol\,/\,L$ non-hdl cholesterol, is $0.7\%$~(female) and $1.3\%$~(male), not zero. However, training the linear regression models with intercept led, in our case, to negative intercepts which  (1) clearly undermines the interpretability of the models (what does a negative 10-year CVD risk mean?) and (2) is a worse point estimate than a zero intercept when compared to the actual SCORE2 values ($0.7\%$ and  $1.3\%$). Given these points, we decided to use a no-intercept linear regression model. 
\end{enumerate}

\begin{table}[h!t]
    \centering
        \caption{SCORE2 risk factor ranges. Ranges are derived from ESC SCORE2 graphical score charts and Risk Calculation App valid ranges.}
    \begin{tabular}{l|r|r}
         Risk Factor & Min & Max \\
         \toprule
         Age ($years$)\textsuperscript{*} & $45$ & $70$ \\
         Systolic Blood Pressure ($mmHg$) & $100$ & $180$ \\
         Non-HDL Cholesterol ($mmol\,/\,L$)& $3$ & $7$ \\
         Total Cholesterol ($mmol\,/\,L$)& $3$ & $9$ \\
         HDL Cholesterol ($mmol\,/\,L$)& $0.7$ & $2.5$ \\
    \end{tabular}
    \vspace{-2mm}
    \begin{minipage}{0.95\linewidth}
    \textsuperscript{*}\textit{Minimum age derived from ESC is 40 years, but we raised it to 45 due to the lack of patients from 40 to 45 in our dataset.}
    \end{minipage}
    \label{tab:rangetable}
\end{table}
\subsection{Framingham Heart Study dataset} \label{dataset}
The patient data we used to design and develop the visualization to communicate CVD risk comes from the Framingham Heart Study `teaching' dataset provided by the Biologic Specimen and Data Repository Information Coordinating Center~\cite{giffen2015}. This version is an anonymized subset version of the data originally collected from the Framingham Heart Study, ensuring the protection of patient confidentiality. The dataset contains data from 4434 individuals, including laboratory, clinic, questionnaire, and information on the following events: angina pectoris, myocardial infarction, atherothrombotic, infarction, cerebral hemorrhage (stroke), or death. Participants were followed for a total of 24 years and clinic data was collected throughout three assessment periods.

To be able to use the 10-year CVD risk predictions of the SCORE2 model, we selected the subset of individuals available during the third period, as it is the only period in which HDL cholesterol was measured in addition to total cholesterol. Of the 3263 individuals for whom clinical data are available during the third period, we filtered those with previous CVD or diabetes (596 individuals), as they are not the target population of the SCORE2 model. Finally, we removed those with missing values or values outside the range indicated in \autoref{tab:rangetable} for at least one of the inputs required by the SCORE2 model, ending with a total of 1762 individuals (713 male and 1049 female).

The individuals in the dataset are located in Framingham, Massachusetts (United States of America). Since the WHO age-standardized CVD mortality rate per 100,000 population in the United States of America for 2019 is $128.1$~\cite{world2019}, the corresponding SCORE2 CVD region risk level is \textit{moderate} (100 to 150 CVD deaths per 100,000\cite{SCORE22021}).
Therefore, for the remainder of the text, we will refer to the SCORE2 model for populations at moderate CVD risk simply as the SCORE2 model.

\subsection{Validation}\label{validation}
Training the no-intercept linear regression shown in \autoref{eq:surrogate} with patient data described in Section~\ref{dataset} and the output of the sex-specific SCORE2 models for the moderate risk region leads to the following two global surrogates models:
\begin{subequations}\label{eq:models}
    \begin{flalign}
     &\hspace*{-11pt}\text{\textit{Male Global Surrogate}} \nonumber \\
     &\hspace*{-11pt}y = 0.087 \cdot z_{age} + 0.058 \cdot z_{sbp} +  0.037 \cdot z_{smoking} +  0.022 \cdot z_{nonhdl} \label{seq:male} \\[5pt] 
     &\hspace*{-11pt}\text{\textit{Female Global Surrogate}} \nonumber \\
     &\hspace*{-11pt}y = 0.060 \cdot z_{age} + 0.037 \cdot z_{sbp} +  0.025 \cdot z_{smoking} +  0.004 \cdot z_{nonhdl} \label{seq:female}
    \end{flalign}
\end{subequations}

To ensure that the design decisions made led to a good trade-off between fidelity and interpretability, we evaluated the fidelity of these two surrogate models to the original SCORE2 models. Reporting on the fidelity of the surrogate model is also one of the items on the checklist for the assessment of medical AI studies provided by Cabitza \& Campagner~\cite{cabitza2021}.

There are two major decisions to consider when evaluating fidelity: selecting a fidelity metric and deciding the fidelity required for a surrogate to be deemed reliable. 
We use the Spearman's correlation (Sp. Corr.) between the predictions of the original model and the global surrogate as it has been proposed as the most appropriate fidelity metric~\cite{schwartzenberg2020}. 
To determine whether the models' fidelity is sufficient, we compared the fidelity metrics of our no-intercept linear regression surrogate models to those of the SCORE2 graphical score charts on the dataset described in Section~\ref{dataset}.

The results showed that the Spearman's correlation of our male (\mbox{Sp. Corr. $= 0.969$}) and female (\mbox{Sp. Corr. $= 0.964$}) linear regression global surrogate models is slightly higher than that of the male (\mbox{Sp. Corr. $= 0.947$}) and female (\mbox{Sp. Corr. $= 0.949$}) SCORE2 graphical score chart. Therefore, the fidelity of both linear regression surrogate models was deemed adequate.
\section{Petal Product Plots} \label{section:ppp}
In this section, we describe the design process and characterize a novel visualization, Petal Product Plots (PPPs). PPPs visually represent the dot product of two non-negative $n$-dimensional vectors $\mathbf{b}$ and $\mathbf{z}$, i.e.,
\begin{equation}\label{eq:dotproduct}
\mathbf{b} \cdot \mathbf{z}=  \sum_{i=1}^ {n}  b_ {i} z_ {i}  =  b_{1} z_{1} + b_{2} z_{2} + \cdots + b_{n} z_{n}\,, \quad b_i \geq 0,\, z_i \geq 0\,.
\end{equation}
The prior data abstraction clearly represents the two no-intercept linear regression models shown in \autoref{eq:models}, where $\mathbf{b} \cdot \mathbf{z}$ is the model output, $b_i$ is the model coefficient for each risk factor, and $z_i$ is the patient's normalized value for the same risk factor.

\subsection{Motivation and design goals}
The design of Petal Product Plots was inspired by the work of Jeong~et~al.~\cite{jeong2014}, which proposes the use of a novel weighted radar chart to communicate metabolic syndrome risk factors’ importance and patient values. Weighted radar charts are a variation of radar charts in which the angles between the different dimension axes are not equal but instead encode a given weight. In the metabolic syndrome use case, each risk factor value is mapped to a dimension axis of the radar chart, and the importance of each risk factor is encoded as an angle between the corresponding risk factor axis and the next anticlockwise axis. One limitation of this design discussed by Jeong~et~al. is that the resulting shape area, which is used as a metabolic syndrome risk quantification, is heavily dependent on the order of the axes. Moreover, it is challenging to identify which weight corresponds to each risk factor, as it is encoded between two risk factor axes.

We drew inspiration from two existing works when designing our solution to these limitations. The first was the characterization of product plots by Hadley \& Hoffmann \cite{hadley2011}. Product plots are visualizations in which the area of the graphical elements is proportional to an underlying attribute product of two other encoded values. The second was a study conducted by Albo et al. \cite{albo2016} that looked at user preferences between radar charts and a flower-based visualization. Their findings revealed a strong preference for flower-based visualization, which was found to be more enjoyable, clear and understandable, easy to use, and easy to learn.

In an effort to design a visualization that combines the appeal of the flower-based visualization and the encoding properties of product plots, we derived the following design goals.
\begin{itemize}[noitemsep]
    \item[\textbf{G1}.]\textbf{Resemble a flower.} The visualization consists of $n$ equally shaped graphical primitive elements that we refer to as \textit{petals}, each representing a summand $b_ {i} z_ {i}$. The petals are arranged by rotating them around a point (the ``flower'' center), fully covering the 360 degrees.
    \item[\textbf{G2}.]\textbf{Petal area is proportional to the $b_i z_i$ product.} The area of $i$-th petal should be proportional to $b_i z_i$. 
    \item[\textbf{G3}.] \textbf{Non-overlapping petals.} As in the product plots constraints established by Hadley \& Hoffmann \cite{hadley2011}, we require the petals to be disjoint so that the entire area is observable. 
    \item[\textbf{G4}.]\textbf{The petal angle encodes the $\bm{b_i}$ value.} The need to have non-overlapping petals fully covering the 360 degrees (or $2\pi$ radians) to resemble a flower, requires the $\ell_1$-normalization of $\mathbf{b}$ within the angle-encoding transformation. Thus, the petal angle in radians is given by $2\pi\cdot b_i/\lVert\mathbf{b}\rVert_1$.
    \item[\textbf{G5}.]\textbf{The value $z_i$ is mapped into the petal length}. As in the radial product plots analyzed by Hadley \& Hoffmann \cite{hadley2011}, to keep the product proportional to the area (design goal~G2), the value mapped to the radius must be square-root transformed. Therefore, the length of the petal must be proportional to $\sqrt{z_i}$.

\end{itemize}
\subsection{Petal shape}
For the shape of the petal, we based ourselves on the rhodonea curves, also known as rose curves, which were described by Grandi~\cite{grandi1723} in the 18th century. Rhodonea curves are given by the polar equation
\begin{equation}\label{eq:rhodonea}
    r = a\cdot\cos(m\theta)\quad\text{or}\quad r = a\cdot\sin(m\theta)\,,
\end{equation}
where $a$ represents the length of the rhodonea petals and each petal is described by a sinusoid half-cycle of $\pi/m$ radians. 

Since the two formulations listed in \autoref{eq:rhodonea} differ only by a rotation around the origin, we will adopt the sine formulation from now on, as it has one entire petal bounded by the angular interval ${0\leq\theta\leq\pi/m}$, which simplifies the notation. In particular, the petal of the rhodonea curve is bounded by an angle of $\beta$ radians when ${m=\pi/\beta}$. Thus, we define a \textit{rhodonea petal} of angle $\beta$ and length $a$ as
\begin{equation*}
    \altmathcal{R}_{\beta,\,a}=\left\{(r,\,\theta):\;\; 0 \leq r \leq a\cdot\sin\left(\frac{\pi}{\beta}\theta\right)\,,\;\; 0\leq\theta\leq\beta\right\}
\end{equation*}
The value $a$ of the rhodonea petal $\altmathcal{R}_{\beta,\,a}$ is expressed as the length of the petal, a clearly observable encoding. In contrast, the value $\beta$ is only implicitly expressed as the bounding angle of the petal, as other common observable encodings of the angle (for example, the area, arc or chord of pie charts slices~\cite{kosara2022}) are also affected in this case by the length of the petal $a$. Therefore, to explicitly encode the angle $\beta$ and thus facilitate its observation, we designed a new family of petal shapes that result from the affine combination of the rhodonea petal $\altmathcal{R}_{\beta,\,a}$ and the circumference arc $r=a$ for $0\leq\theta\leq\beta$. Thus, for each $0<\kappa<1$, we define our proposed petals of angle $\beta$ and length $a$ as
\begin{equation*}
    \hspace*{-8pt}\altmathcal{P}_{\kappa,\,\beta,\,a}=\left\{(r,\,\theta):\;\; 0 \leq r \leq \kappa \cdot a + (1-\kappa)\cdot a\cdot\sin\left(\frac{\pi}{\beta}\theta\right)\,,\;\; 0\leq\theta\leq\beta\right\}
\end{equation*}

A comparison of the rhodonea petal and our proposed petal for $\kappa=0.5$ can be seen in \autoref{fig:petal}, where it can be observed that the angle $\beta$ of our proposed petal $\altmathcal{P}_{\kappa,\,\beta,\,a}$ can now be observed through the central angle encoding, similar to the central angle of a slice in a pie chart. 

\begin{figure}[h!t]
\centering
\includegraphics[width=0.8\linewidth]{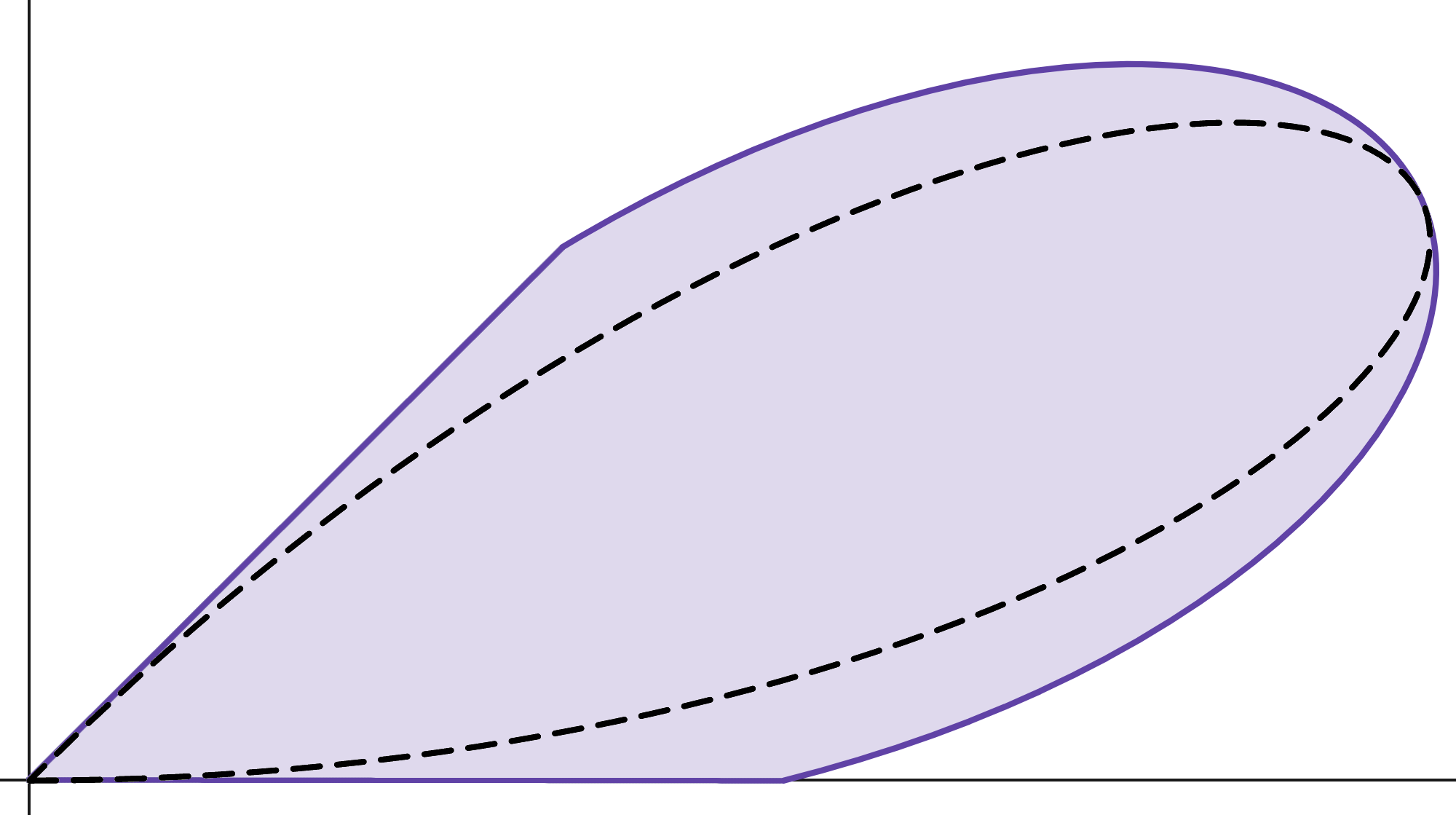}
\caption{Rhodonea petal $\altmathcal{R}_{\beta,\,a}$ (dashed line) on top of our proposed petal shape $\altmathcal{P}_{\kappa,\,\beta,\,a}$ (solid line) for $a=1$, $\beta=\pi/4$, $\kappa=0.5$.} 
\label{fig:petal}
\end{figure}

Finally, the following result shows that the area of the petal $\altmathcal{P}_{\kappa,\,\beta,\,a}$ is proportional to the squared of its length times its angle, i.e.,~$a^2\beta$.

\begin{proposition} \label{proposition}
Given $\kappa$, $0<\kappa<1$, the area under the polar curve
\[
r = \kappa \cdot a + (1-\kappa)\cdot a\cdot\sin\left(\frac{\pi}{\beta}\theta\right)
\] between $\theta=0$ and $\theta=\beta$ is given by  
\begin{equation*}
A=\beta\;a^2 K\,,\quad\text{where}\quad K = \frac{-8\kappa^2+8\kappa+\pi(3\kappa^2-2\kappa+1)}{4\pi}\,.
\end{equation*}

\end{proposition}
\begin{proof}
The area under a polar curve between $\theta=s$ and $\theta=t$ is given by 
\[
\int_s^t \frac{1}{2}r^2 d\theta
\]
Thus, the area of $\displaystyle r = \kappa \cdot a + (1-\kappa)\cdot a\cdot\sin\left(\frac{\pi}{\beta}\theta\right)$ between $\theta=0$ and $\theta=\beta$ is given by
\begin{align*}
A &= \int_0^\beta \frac{1}{2}\left(\kappa \cdot a + (1-\kappa)\cdot a\cdot\sin\left(\frac{\pi}{\beta}\theta\right)\right)^2 d\theta \\
&= \beta a^2 \left(\frac{-8\kappa^2+8\kappa+\pi(3\kappa^2-2\kappa+1)}{4\pi}\right) = \beta a^2 K\,.
\end{align*}

\end{proof}

\subsection{Petal composition}
Given the $\mathbf{b} \cdot \mathbf{z}$ dot product described in \autoref{eq:dotproduct}, if we define the angle and length of the $i$-th petal $\altmathcal{P}_{\kappa,\,\beta_i,\,a_i}$ as
\begin{equation}\label{eq:angleandlength}
\beta_i= 2\pi\frac{b_i}{\lVert\mathbf{b}\rVert_1}\quad\text{and}\quad a_i=\sqrt{z_i}\,,
\end{equation}

the petal angle encodes the $b_i$ value (design goal~G4) and the value $z_i$ is mapped square-root transformed into the petal length (design goal~G5). Additionally, since the sum of the angles $\beta_i$ is equal to $2\pi$ radians, we can arrange the petals without overlapping  (design goal~G3), by adequately rotating the petals around the origin. Finally, the area of the $i$-th petal $\altmathcal{P}_{\kappa,\,\beta_i,\,a_i}$ is proportional to the $b_ {i} z_{i}$ product (design goal~G2), as shown in the following result. 
\begin{corollary}
Given two non-negative $n$-dimensional vectors $\mathbf{b}$~and~$\mathbf{z}$ and a scalar $\kappa$, with $0<\kappa<1$, the area of each $i$-th petal $\altmathcal{P}_{\kappa,\,\beta_i,\,a_i}$, where
\[
\beta_i=2\pi\frac{b_i}{{\lVert\mathbf{b}\rVert_1}}\quad\text{and}\quad a_i=\sqrt{z_i}\,,
\]
is proportional to the product of the $i$-th element of each vector, $b_i z_i$.

\end{corollary}
\begin{proof}
Lets denote the area of the $i$-th petal $\altmathcal{P}_{\kappa,\,\beta_i,\,a_i}$ as $A_i$. By Proposition \ref{proposition}, it follows that
\[
A_i= \beta_i a_i^2 K =  2\pi\frac{b_i}{\lVert\mathbf{b}\rVert_1}(\sqrt{z_i})^2 K=b_iz_i\frac{2\pi K}{\lVert\mathbf{b}\rVert_1}\,,
\]
where
\[
K = \frac{-8\kappa^2+8\kappa+\pi(3\kappa^2-2\kappa+1)}{4\pi}\,.
\]
Therefore, the area $A_i$ is proportional to the product $b_ {i} z_ {i}$, since $2\pi K/\lVert\mathbf{b}\rVert_1$ is a constant which value depends solely on the $\ell_1$-norm of the given vector $\mathbf{b}$ and the fixed parameter~$\kappa$.
\end{proof}

Therefore, the proposed visualization fulfills G2, G3, G4, and G5 design goals. However, we consider the resemblance to a flower (design goal~G1) somewhat limited since, as shown in \autoref{fig:noflower}, the resulting graphical elements for different angle values do not seem equally shaped.
This is the case even when, as in \autoref{fig:noflower}, the petal length is the same for each petal.

\begin{figure}[h!t]
     \centering
     \begin{subfigure}[t]{0.49\linewidth}
         \centering
         \includegraphics[width=\textwidth]{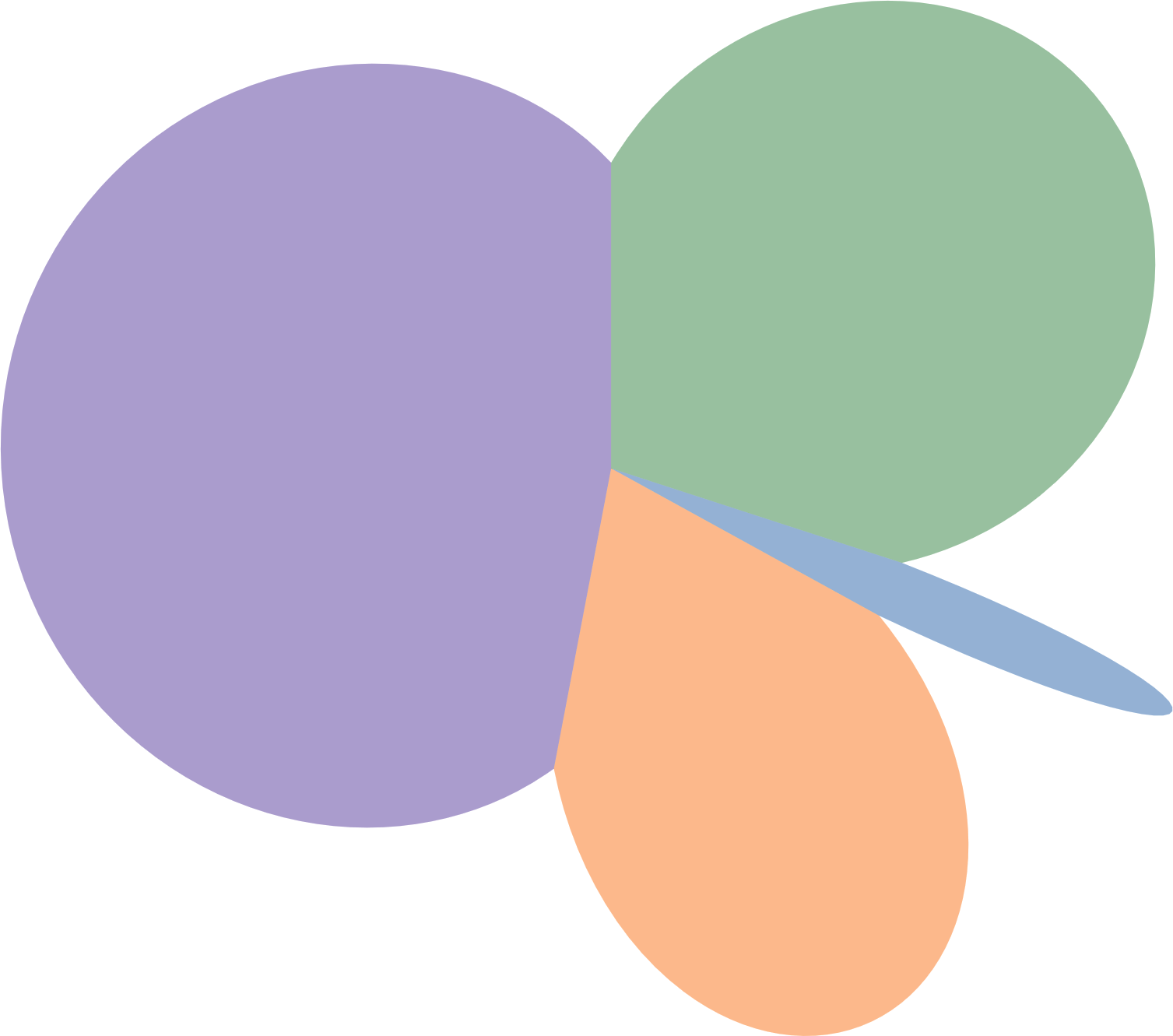}
         \caption{$\altmathcal{P}_{\kappa,\,\beta_i,\,a}$ petals. The angles $\beta_i$ are, from the largest to the smallest, $2\pi\cdot 0.47$ (purple), $2\pi\cdot 0.3$ (green), $2\pi\cdot 0.2$ (orange) and $2\pi\cdot 0.03$ (blue).}
         \label{fig:noflower}
     \end{subfigure}
     \hfill
     \begin{subfigure}[t]{0.45\linewidth}
         \centering
         \includegraphics[width=\textwidth]{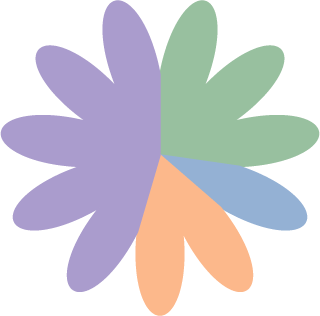}
         \caption{$\altmathcal{M}_{N, \kappa,\,\eta,\,a}$ multilobe petals with $N = 11$. The number of lobes $\eta_i$ are, from the largest to the smallest, $5$ (purple), $3$ (green), $2$ (orange) and $1$ (blue).}
         \label{fig:flower}
     \end{subfigure}
        \caption{Four petals rotated around the origin to fully cover 360 degrees with $\kappa=0.5$ and the same length for every petal ($a=1$).} 
        \label{fig:shapes}
\end{figure}

The solution to this lack of resemblance with a flower was inspired by nature itself: multilobed petals.

\subsubsection{Multilobed petal composition}
Lobes are the upper (free) parts of the petals, the number of which can vary between species. For example, \textit{Cerastium arvense} petals are bifid (2 lobes), \textit{Layia glandulosa} petals are trifid (3 lobes), and \textit{Silene hookeri} petals have 4 lobes. 

In our case, we allow each of our $n$ petals to have a different number of lobes $\eta_i$. More precisely, each of our $n$ multilobed petals of length $a_i$ and angle $\gamma_i=\eta_i\beta^\ast$ results from stitching together $\eta_i$ equal $\altmathcal{P}_{\kappa,\,\beta^\ast,\,a_i}$ shapes. The angle $\beta^\ast$ spanned by each lobe is kept constant to homogenize the shape of our $n$ multilobed petals. Furthermore, to arrange the multilobed petals resembling a flower (design goal~G1) and without overlapping (design goal~G3) as shown in \autoref{fig:flower}, the sum of their angles $\gamma_i=\eta_i\beta^\ast$ must be $2\pi$ radians; therefore, 

\[
\beta^\ast=\frac{2\pi}{\sum_{i=1}^n\eta_i}\quad\text{and}\quad\gamma_i=\eta_i\beta^\ast=2\pi\frac{\eta_i}{\sum_{i=1}^n\eta_i} 
\]

Thus, the shape of our multilobed petals is univocally determined by the overall number of lobes $N=\sum_{i=1}^n\eta_i$, the petal number of lobes $\eta_i$, the petal length $a_i$, and the parameter $\kappa$. Accordingly, we denote our multilobed petals by $\altmathcal{M}_{N,\kappa,\,\eta,\,\,a}$. The length of a multilobed petal $\altmathcal{M}_{N,\kappa,\,\eta,\,\,a}$ is given by $a$, its angle by $\gamma=2\pi\;\eta/N$  and, as shown in the following result, its area is the same as that of a unilobed petal with the same length and angle, i.e., a petal $\altmathcal{P}_{\kappa,\,\gamma,\,a}$. 

\begin{corollary}\label{cor:multilobearea}
The area of a multilobe petal, $\altmathcal{M}_{N,\kappa,\,\eta,\,\,a}$, where $N$ an $\eta$ are non-negative integers with $\eta \leq N$, $0<\kappa < 1$, and $a\geq 0$ is given by 
\[
A= a^2\gamma\; K\,,
\]
where
\[ 
\gamma=2\pi\frac{\eta}{N}\quad\text{and}\quad K = \frac{-8\kappa^2+8\kappa+\pi(3\kappa^2-2\kappa+1)}{4\pi}\,.
\]
\end{corollary}
\begin{proof}
The area of a multilobe petal $\altmathcal{M}_{N,\kappa,\,\eta,\,\,a}$ is, by design, $\eta$ times the area of a $\altmathcal{P}_{\kappa,\,\beta^\ast,\,a}$ petal, where
\[
\beta^\ast=\frac{2\pi}{N}\,.
\]
By Proposition \ref{proposition}, it follows that the area is thus given by
\[
\eta\cdot\left(a^2\beta^\ast K \right)= a^2 \frac{2\pi\eta}{N}  K\,.
\]
  
\end{proof}

Given the $\mathbf{b} \cdot \mathbf{z}$ dot product described in \autoref{eq:dotproduct}, we can map the value $z_i$ square-root transformed into the length of the $i$-th multilobe petal $\altmathcal{M}_{N,\kappa,\,\eta_i,\,\,a_i}$ (design goal~G5) by setting $a_i=\sqrt{z_i}$. However, since the number of petals $\eta_i$ must be a non-negative integer, we can only set $\eta_i=b_i$ and thus encode precisely $b_i$ in the multilobe petal angle $\gamma_i$ (design goal~G4) if the elements of $\mathbf{b}$ are non-negative integers. In that case, the area of the $i$-th multilobe petal $\altmathcal{M}_{N,\kappa,\,\eta_i,\,\,a_i}$ is also proportional to the $b_iz_i$ product (design goal~G2), as shown in the following result.

\begin{corollary}
Given two non-negative $n$-dimensional vectors $\mathbf{b}$~and~$\mathbf{z}$ and a scalar $\kappa$, with $0<\kappa<1$. If the elements of $\mathbf{b}$ are all non-negative integers, then the area of each $i$-th multilobe petal $\altmathcal{M}_{N,\kappa,\,\eta_i,\,\,a_i}$, where $N=\sum_{i=1}^n\eta_i$, $\eta_i=b_i$ and $a_i=\sqrt{z_i}$, is proportional to the product of the $i$-th element of each vector, $b_i z_i$.
\end{corollary}
\begin{proof}
Lets denote the area of the $i$-th petal $\altmathcal{M}_{N,\kappa,\,\eta_i,\,\,a_i}$ as $A_i$. Since $\eta_i=b_i$, the angle $\gamma_i$ of the the $i$-th petal $\altmathcal{M}_{N,\kappa,\,\eta_i,\,\,a_i}$ is given by $2\pi b_i/N$. By Corollary \ref{cor:multilobearea}, it follows that
\[
A_i= a^2 \gamma_i  K = (\sqrt{z_i})^2 2\pi\frac{b_i}{\lVert\mathbf{b}\rVert_1}K=b_iz_i\frac{2\pi K}{\lVert\mathbf{b}\rVert_1}\,,
\]
where
\[
K = \frac{-8\kappa^2+8\kappa+\pi(3\kappa^2-2\kappa+1)}{4\pi}\,.
\]
Therefore, the area $A_i$ is proportional to the product $b_ {i} z_ {i}$, since $2\pi K/\lVert\mathbf{b}\rVert_1$ is a constant which value depends solely on the $\ell_1$-norm of the given vector $\mathbf{b}$ and the fixed parameter~$\kappa$.
\end{proof}

Otherwise, i.e., if there exists an $i$ such that $b_i$ is not a non-negative integer, the following result shows that for a given overall number of lobes $N$, Hamilton's apportionment method provides us with the $\eta_i$ non-negative integers that minimize the sum of absolute deviations between the multilobes petal angle $\gamma_i$ and the perfect proportional angles $\beta_i$ from \autoref{eq:angleandlength}.
Furthermore, we show that the approximation error for each angle is bounded $2\pi/N$ and, thus, can be reduced by increasing the overall number of lobes~$N$.

\begin{proposition}\label{proposition2}
Let $\mathbf{b}$ be a non-negative $n$-dimensional vector and $N$ be non-negative integer. If we denote
\[
q_i=N\frac{b_i}{\lVert\mathbf{b}\rVert_1} \quad\text{and}\quad \mathbf{q} = \left( q_1,\cdots,q_n \right)\,,
\]
then any of the one or more solutions $\bm{\eta}=\left(\eta_1,\cdots,\eta_n \right)=f(\mathbf{q},N)$ of Hamilton's apportionment method, $M(\mathbf{q},N)$, is also a solution of the following minimization problem
\[
\underset{\bm{\eta}\,\in\,\mathbb{N}^n}{\text{minimize}}\quad\sum_{i=1}^n  \left\lvert\,2\pi\frac{\eta_i}{\sum_{i=1}^n\eta_i} - 2\pi\frac{b_i}{\lVert\mathbf{b}\rVert_1}\,\right\rvert\quad\text{subject to}\quad \sum_{i=1}^n\eta_i=N\,.
\]
Moreover, if we denote
\[
\gamma_i=2\pi\frac{\eta_i}{N}\quad\text{and}\quad\beta_i= 2\pi\frac{b_i}{\lVert\mathbf{b}\rVert_1}\,,
\]
then
\[
\left\lvert\,\gamma_i - \beta_i\,\right\rvert<\frac{2\pi}{N}\,.
\]
\end{proposition}
\begin{proof}
Let $\bm{\eta}=\left(\eta_1,\cdots,\eta_n \right)\,\in\,\mathbb{N}^n$, $\sum_{i=1}^n\eta_i=N$ be a solution of Hamilton's apportionment method, i.e.,
$f(\mathbf{q},N)=\bm{\eta}\in M(\mathbf{q},N)$. By Garrett \cite{garrett1976}, it follows that $\bm{\eta}$ is the solution of the following minimimization problem
\[
\underset{\bm{\eta}\,\in\,\mathbb{N}^n}{\text{minimize}}\quad\sum_{i=1}^n  \left\lvert\,\eta_i - q_i\,\right\rvert\quad\text{subject to}\quad \sum_{i=1}^n\eta_i=N\,.
\]
Thus, for all $\bm{\eta^\prime}\,\in\,\mathbb{N}^n$, $\displaystyle \sum_{i=1}^n\eta_i^\prime=N$,
\[
\sum_{i=1}^n  \left\lvert\,\eta_i - q_i\,\right\rvert\leq\sum_{i=1}^n  \left\lvert\,\eta_i^\prime - q_i\,\right\rvert\,,
\]
and, multiplying both sides by the positive number $2\pi/N$ and then replacing $q_i$ by $\displaystyle N\frac{b_i}{\lVert\mathbf{b}\rVert_1}$, and $N$ by $\displaystyle \sum_{i=1}^n\eta_i$ and $\displaystyle \sum_{i=1}^n\eta_i^\prime$, we get
\[
\sum_{i=1}^n  \left\lvert\,2\pi\frac{\eta_i}{\sum_{i=1}^n\eta_i} - 2\pi\frac{b_i}{\lVert\mathbf{b}\rVert_1}\,\right\rvert\leq\sum_{i=1}^n  \left\lvert\,2\pi\frac{\eta_i^\prime}{\sum_{i=1}^n\eta_i^\prime} - 2\pi\frac{b_i}{\lVert\mathbf{b}\rVert_1}\,\right\rvert\,.
\]
Therefore, $\bm{\eta}$ is also a solution of the desired minimization problem.

Additionally, since Hamilton's apportionment method \textit{stays within the quota} \cite{garrett1976}, it follows that
\[
\left\lvert\,\eta_i - q_i\,\right\rvert< 1\quad\text{for all}\;i\,.
\]
Again, multiplying both sides by $2\pi/N$ and then replacing $q_i$ by $\displaystyle N\frac{b_i}{\lVert\mathbf{b}\rVert_1}$ we get the desired bound
\[
\left\lvert\,2\pi\frac{\eta_i}{N} - 2\pi\frac{b_i}{\lVert\mathbf{b}\rVert_1}\,\right\rvert<\frac{2\pi}{N}\quad\text{for all}\;i\,.
\]
\end{proof}
Finally, it is important to highlight that the multilobe petal angle being $\gamma_i=\eta_i\beta^\ast$ also provides a new observable encoding alternative to the central angle: the number of lobes of the petal, $\eta_i$. This is a notable advantage of multilobed petals since prior work by Skau \& Kosara~\cite{skau2016} found the central angle the least likely visual cue used for reading in visualizations with similar angular encodings such as pie charts. Furthermore, this new encoding allows users to process information using counting strategies similar to those used on icon arrays~\cite{kreuzmair2016}.
\subsection{Petal Product Plots remarks}
To sum up, the proposed Petal Product Plots $PPP_{N,\kappa}$ represent the dot product of two non-negative $n$-dimensional vectors $\mathbf{b}$ and $\mathbf{z}$, i.e.,
\begin{equation}\label{eq:remarks}
\mathbf{b} \cdot \mathbf{z}=  \sum_{i=1}^ {n}  b_ {i} z_ {i}  =  b_{1} z_{1} + b_{2} z_{2} + \cdots + b_{n} z_{n}\,, \quad b_i \geq 0,\, z_i \geq 0
\end{equation}
In addition to the input data $\mathbf{b}$ and $\mathbf{z}$, the following two parameters determine the visualization graphical elements.
\begin{itemize}[noitemsep]
    \item The number of lobes $N$. As shown in Proposition \ref{proposition2}, this parameter also determines the precision with which each $b_i$ can be encoded. 
    \item The parameter $\kappa$, where (1-$\kappa$) is the proportion of the length of each petal corresponding to the lobes, i.e., the upper free parts of the petals.
\end{itemize}
It is important to note that although $\mathbf{b}$ and $\mathbf{z}$ have interchangeable roles in \autoref{eq:remarks}, their different encodings (petal angles and lengths) make them singular due to $\ell_1$-normalization of $\mathbf{b}$ within the angle-encoding transformation.
Therefore, we believe that this makes PPPs especially suitable for scenarios where the values of the $b_i$ elements can be understood as a part-to-whole.

\section{Petal-X: PPP for CVD risk communication} \label{section:petalx}
Petal-X (see \autoref{fig:teaser}) provides personalized SCORE2 10-year CVD risk explanations using Petal Product Plots to visualize the tailor-made SCORE2 global surrogate models described by \autoref{eq:models}. 

To do so, we encode the linear regression coefficients in the petal angles since coefficients can be interpreted as each risk factor's relative (part-to-whole) contribution. Hence, the patient's (normalized) risk factor values are mapped square-root transformed to the petal length, and the area of the petal represents the risk factor contribution to the 10-year CVD risk. We set a low number of lobes $N$ to facilitate the use of counting strategies. In particular, we set $N=10$ for the male surrogate model and $N=11$ for the female surrogate model. The one-lobe increase in the number of petals in female Petal-X is due to the fact that the optimal Hamilton attribution method would otherwise assign zero lobes to the \textit{non-HDL cholesterol} risk factor. Finally, we set the parameter $\kappa$ to its midrange value (0.5). 

In the remainder of this section, we validate the Petal-X fidelity for the chosen number of lobes $N$ and justify other design decisions.

\subsection{Petal-X fidelity validation}
The selected PPP number of lobes ($N=10$ for the male model and $N=11$ for the female model) led to the encoding of the following approximations of the coefficients of each \autoref{eq:models} surrogate model into the petal angles.
\begin{subequations}\label{eq:petalxmodels}
    \begin{flalign}
     &\hspace*{-11pt}\text{\textit{Petal-X (N=10) Male Global Surrogate}} \nonumber \\
     &\hspace*{-11pt}y =  0.082 \cdot z_{age} + 0.061 \cdot z_{sbp} + 0.041 \cdot z_{smoking} + 0.020 \cdot z_{nonhdl}\label{seq:petlaxmale} \\[5pt] 
     &\hspace*{-11pt}\text{\textit{Petal-X (N=11) Female Global Surrogate}} \nonumber \\
     &\hspace*{-11pt}y =  0.057 \cdot z_{age} + 0.034 \cdot z_{sbp} + 0.023 \cdot z_{smoking} + 0.011 \cdot z_{nonhdl} \label{seq:petalxfemale}
    \end{flalign}
\end{subequations}

Describing the global linear regression surrogate models actually represented in Petal-X allows us to evaluate the fidelity of Petal-X to the original SCORE2 models. 
As in Section~\ref{validation}, we compared the fidelity of the Petal-X and SCORE2 graphical score charts on the dataset described in Section~\ref{dataset}. For a more complete analysis, in addition to Spearman's correlation, three other common fidelity metrics were included: coefficient of determination ($R^2$, calculated as the squared Pearson correlation coefficient), root mean squared error (RSME), and mean absolute error (MAE). 

\begin{table}[h!t]
\caption{Fidelity metrics of the male and female linear regression global surrogate models, Petal-X approximations, and SCORE2 graphical score charts. The best score for each sex and metric is highlighted in bold.}
\label{tab:viz_fidelity}
\begin{tabular}{@{}rrrrr@{}}
\toprule
\textbf{Male} Surrogate Model & Sp. Corr. & $R^{2}$ & RSME & MAE \\ \midrule
Graphical Score Chart                         &0.947&0.907&\textbf{0.013}&\textbf{0.010}  \\
Linear Regression (Eq.\ref{seq:male})         &\textbf{0.969}&\textbf{0.917}&0.015&0.011\\ 
Petal-X (N=10) (Eq.\ref{seq:petlaxmale})            &0.962&0.906&0.015&0.011\\ \bottomrule
\toprule
\textbf{Female} Surrogate Model & Sp. Corr. & $R^{2}$ & RSME & MAE \\ \midrule
Graphical Score Chart                            &0.949&0.897&\textbf{0.009}&\textbf{0.007}\\
Linear Regression (Eq.\ref{seq:female})         &0.964&0.890&0.011&0.009\\
Petal-X (N=11) (Eq.\ref{seq:petalxfemale})            &\textbf{0.972}&\textbf{0.903}&0.012&0.009\\ \bottomrule
\end{tabular}
\end{table}

The fidelity of the male and female linear regression surrogate models represented by Petal-X was deemed adequate, as their Spearman correlation is slightly higher than that of the SCORE2 graphical score charts for each sex, and there are no notable differences in any of the other fidelity metrics (see \autoref{tab:viz_fidelity}).

\subsection{Design decisions}
In addition to using PPPs, other design decisions were made to improve the comprehensibility of Petal-X visual explanations. Some of these decisions were influenced by ESC SCORE2 graphical score charts in order to (1) achieve external consistency and (2) avoid confounding design elements in Petal-X evaluation, since it uses SCORE2 GSCs for the control group.

\textbf{Color.} We chose to use one color for the entire Petal-X to give users a quick visual notice of the severity of the 10-year CVD risk. We use the age-dependent color scale used by ESC GSCs and mobile app, with

\noindent \tikz\draw[redrisk,fill=redrisk] (0,0) circle (.6ex); Red for \textit{very high CVD risk} (treatment recommended)\\
\noindent \tikz\draw[orangerisk,fill=orangerisk] (0,0) circle (.6ex); Orange for \textit{high CVD risk} (treatment to be considered)\\
\noindent \tikz\draw[greenrisk,fill=greenrisk] (0,0) circle (.6ex); Green for \textit{low-to-moderate CVD risk} (treatment not recommended).

\textbf{Grid.} We use a petal-shaped grid design with dotted lines at $0$, $0.25$, $0.5$, $0.75$, and $1$ (only $0$ and $1$ for smoking) of the normalized risk factor values. In addition, grid lines are labeled with the denormalized risk factor values. Both decisions help users understand how a petal area changes if the value of a risk factor is modified, i.e., project risk reduction through risk factor modification. 

\textbf{Legends and labels.} In addition to the color legend, a lobe risk legend with the same levels as the grid is provided to help users identify the risk that each area represents. Furthermore, each petal is labeled outside the grid with the risk factor name and the patient numerical value. In this version, neither the total CVD risk nor the individual risk factor contributions are provided numerically to be able to understand the ease of use of the area encoding on the Petal-X evaluation. 

\section{Evaluation}\label{section:evaluation}
We evaluated Petal-X by conducting a controlled experiment with healthcare experts, as they play a critical role in determining whether or not to adopt Petal-X in clinical settings. The focus on individuals with healthcare knowledge also limits the likelihood of clinical misconceptions arising from participation in the evaluation. 
Before conducting the study, ethical clearance was obtained from the Commission on Ethical Issues in Nursing of the Faculty of Health Sciences (University of Maribor).

The main objective of this evaluation is to compare the Petal-X design with a more traditional representations of CVD risk predictions as graphical score charts in terms of (1) support for the three challenges devised in Section \ref{sec:challenges} and (2) perceived transparency, trust, and intent to use. 

\subsection{Study design}\label{studydesign}
We conducted a between-subject design with two conditions: GSC (control group) and Petal-X (experimental group). The study was implemented using the open source \textit{formr} survey framework~\cite{arslan2020}. In both conditions, the experiment was structured into the following phases. 
\begin{itemize}[noitemsep]
\item[\textbf{P0.}]\textbf{Consent form and randomization.} Information regarding the user study is provided to potential participants. If individuals agree to participate in the user study, they are randomly assigned to the control group (GSC) or the experimental group (Petal-X). We use simple randomization to allocate participants as the total number of participants was not known in advance.

\item[\textbf{P1.}]\textbf{Sociodemographics.} We collected information on the participants' age, gender, academic standing, as well as whether or not they had experience with patients with chronic diseases.

\item[\textbf{P2.}]\textbf{Onboarding.} Participants go through an onboarding phase to become familiar with their condition visualization. The same system was implemented for both Petal-X and GSC based on the step-by-step guide design proposed by Stoiber~et~al. \cite{stoiber2022}.

\item[\textbf{P3.}]\textbf{Performance.} We evaluate performance by assessing the participant’s ability to perform three tasks aligned with the three challenges described in Section~\ref{sec:challenges} using their condition visualization (GSC or Petal-X). Each participant runs 4 trials, each based on a different set of clinical data. Data correspond to two patients of each sex randomly chosen from a set of eight. For each trial, the participants had to:
\begin{itemize}[noitemsep]
\item[\textbf{T0.}] Identify (GSC) or calculate (Petal-X) the 10-year CVD risk.
\item[\textbf{T1.}] Determine which modifiable risk factor contributes the most and the least to the patient's 10-year CVD risk. 
\item[\textbf{T2.}] Given two hypothetical scenarios, select the one that would reduce more the patient's 10-year CVD risk.
\end{itemize}
Completion time and error rate are measured for each task. 

\item[\textbf{P4.}]\textbf{Workload.} We evaluated subjective workload by measuring the NASA Task Load Index (NASA-TLX)~\cite{hart2006}. NASA-TLX consists of six subscales: mental demand, physical demand, temporal demand, overall performance, effort, and frustration level. We used an unweighted NASA-TLX questionnaire, as the literature shows similar results ~\cite{hart2006}, with the benefit of time savings.

\item[\textbf{P5.}]\textbf{Transparency, trust, and intent to use.} Participants go through a questionnaire on perceived transparency, trust and intent to use the system. We adapted a questionnaire by Cramer~et~al.~\cite{cramer2008} originally designed for the art recommender systems domain. Based on this questionnaire, we adapted some questions for our application domain and removed questions related to other scales, such as acceptance or competence. \autoref{tab:scales} shows the final scales and questionnaire items. 

\item[\textbf{P6.}]\textbf{Open-ended questions.} We end the study with two open-ended optional questions on study-related issues and suggestions to improve the visualization system. There is also a final open space for additional feedback.
\end{itemize}

\subsection{Participants} \label{sec:participants}
Participant recruitment was carried out at the Faculty of Health Sciences of the University of Maribor in Slovenia. Therefore, the study material, including questionnaires, onboarding and stimuli, was translated into Slovenian. Two people fluent in both English and Slovenian worked independently on the translation. The two translated versions were compared, and the final version was obtained after minor changes. 

The study was carried out for about a month, during which nursing students completed the study individually while in a group class setting. The subject pool consisted of 12 groups with 10-15 participants each. Of these, 124 individuals participated and 107 completed the study. 

To improve reliability, we used the trust scale (see \autoref{tab:scales}) to exclude presumably invalid respondents based on its reversed Likert item rating, adapting the procedure proposed by J{\'o}zsa \& Morgan\cite{jozsa2017}. First, we removed individuals who scored all Likert items on the trust scale equally (including the reversed item before transformation) if the rating was not the neutral middle one (4 on a scale of 1 to 7). Five participants were excluded due to this condition. 
Then, we used an iterative method to remove those respondents whose scores for ``The system is reliable'' and ``The system is deceptive'' (once reversed) differed more than a certain threshold, except if one of the ratings was the neutral middle one. The final threshold was then established based on the included and excluded participants' Cronbach's alpha. This iterative process filtered out 14 participants who agreed (resp. disagree) simultaneously with both ``The system is deceptive'' and ``The system is reliable'' and at least strongly agree (resp. disagree) with one of them. 

Therefore, 88 individuals (82.2\% of those who completed the survey) were finally included in the analysis, 38 of which were assigned to the control group (GSC) and 50 to the experimental group (Petal-X) .All participants were between  21 and 43 years old, with 75\% of them under 23. Most participants were women (85\%) and final-year undergraduate students (89\%). All but one had previous experience with chronic patients.

\subsection{Stimuli and clinical data}
Participants are not asked to input any clinical data throughout the survey. Instead, they are shown the risk factor values and corresponding 10-year SCORE2 CVD risk predictions from a sample of 17 individuals from the anonymized subset data of the Framingham Heart Study described in Section~\ref{dataset}, eight of each sex for the performance phase (P3) and one additional individual for the onboarding phase (P2). For each individual, two hypothetical scenarios are assigned for the task T2 of the performance phase (P3)  based on the treatment goals listed in the 2021 ESC CVD guidelines~\cite{cvdguidelines2021}. In particular, the first two scenarios from the following list that meet the conditions were selected: 
\begin{itemize}[noitemsep]
    \item Lower systolic blood pressure to $130\; mmHg$ (if currently $\geq 140$) or to $110\; mmHg$ (if currently between $120$ and $140$).
    \item Stop smoking (if the patient is currently smoking).
    \item Reduce non-HDL cholesterol to $3.4\;mmol\,/\,L$ (if currently $> 4$).
\end{itemize}

After removing patients who do not meet at least two hypothetical scenario conditions, individuals of each sex were selected using the Kennard \& Stone algorithm~\cite{Kennard1969} available in the \textit{prospectr} package~\cite{stevens2022} in order to obtain a representative sample. 

\subsection{Results} 
In this section, we show the main results of testing the differences between Petal-X and GSC groups. We conducted Welch's $t$-tests as the conditions sample sizes are unequal~\cite{zimmerman2004}.

\begin{figure}[h!t]
\centering
\includegraphics[width=\linewidth]{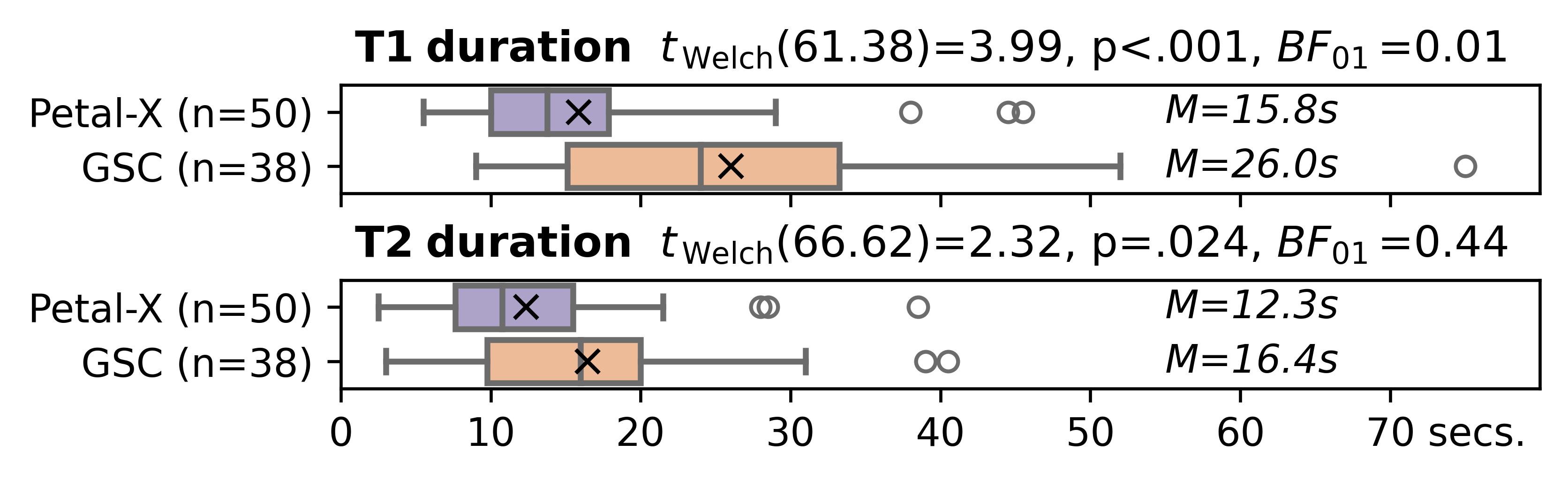}
\includegraphics[width=\linewidth]{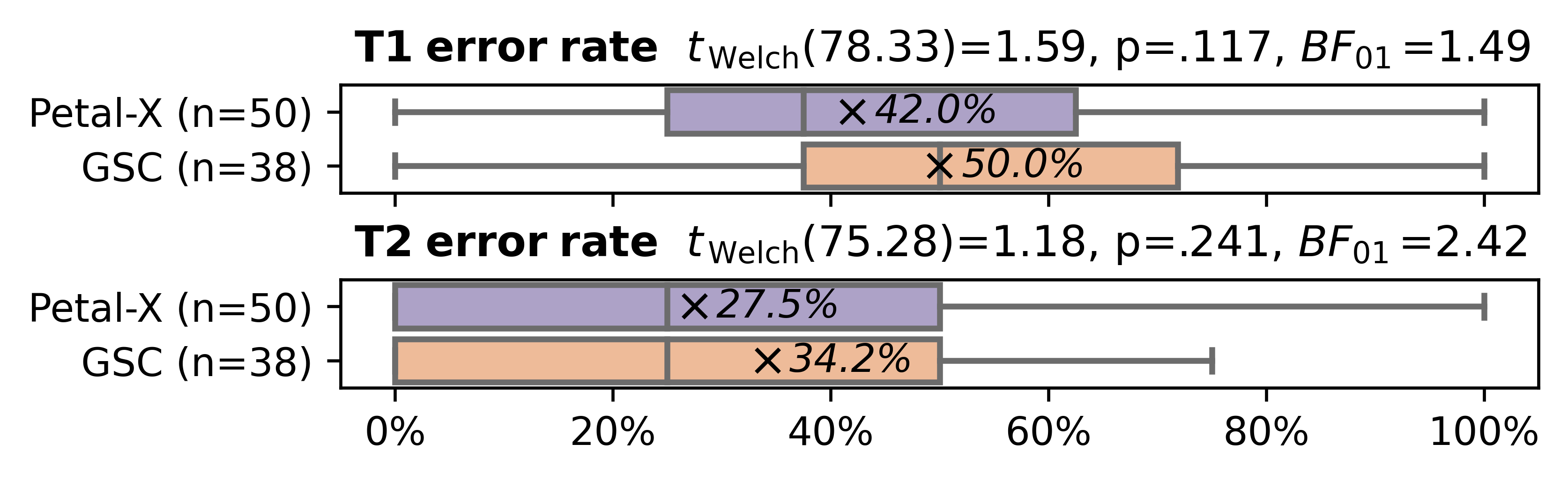}
\caption{Task T1 and task T2 performance results.} 
\label{fig:perfresults}
\end{figure}
\textbf{Task T1 and T2 performance.} The use of Petal-X led to a significantly shorter median duration of tasks T1 and T2 compared to GSC, by 10.2 and 4.1 seconds, respectively. Although the error rates in tasks T1 and T2 were also slightly lower when Petal-X was used compared to GSC, Welch's $t$ tests revealed that these effects were not statistically significant.

\textbf{CVD risk Petal-X calculation (task T0).} The mean difference in the trial-averaged absolute error between the group using Petal-X to calculate the 10-year CVD risk (M=5.37, SD=3.85) and the group merely identifying it in GSC (M=2.74, SD=3.03) is statistically significant ($t_{\mathrm{Welch}}(85.88) = -3.58$, $p<.001$).

\begin{table}[h!t]
\centering
\caption{Transparency, trust, and intent to use scales and items. The questionnaire items are seven-point Likert questions, ranging from 1~(very strongly disagree) to 7 (very strongly agree).}
\label{tab:scales}
\begin{tabular}{@{}p{0.09\linewidth}p{0.85\linewidth}@{}}
\toprule
\multicolumn{2}{l}{[\textbf{Transparency}] Perceived transparency of the system}     \\
\multicolumn{2}{l}{Cronbach's alpha $=0.79$, M $=5.01$, SD$=1.58$, range: $1.0-7.0$} \\[1pt]
\textbf{Tra1} & \textit{I understand why the system predicted the 10-year risk of cardiovascular disease it did.} \\
\textbf{Tra2} & \textit{I understand what the system bases its risk prediction on.}                               \\ \midrule
\multicolumn{2}{l}{[\textbf{Trust}] Trust on the system} \\
\multicolumn{2}{l}{Cronbach's alpha $=0.85$, M $=4.11$, SD$=1.26$, range: $1.0-7.0$} \\[1pt]
\textbf{Tru1} & \textit{I am confident in the system.} \\
\textbf{Tru2R} & \textit{The system is deceptive. (Reversed)} \\
\textbf{Tru3} & \textit{The system is reliable.} \\
\textbf{Tru4} & \textit{I can trust the system.} \\
\textbf{Tru5} & \textit{I can depend on the system.}\\ \midrule
\multicolumn{2}{l}{[\textbf{Intent to use}] Intent to use the system} \\
\multicolumn{2}{l}{Cronbach's alpha $=0.90$, M $=4.11$, SD$=1.88$, range: $1.0-7.0$} \\[1pt]
\textbf{Int1} & \textit{I would like to use the system again for similar tasks.} \\ 
\textbf{Int2} & \textit{The next time I am looking at a risk prediction I would like to use this system.} \\ \bottomrule
\end{tabular}
\end{table}

\textbf{Transparency, trust and intent to use.} We assessed the reliability of the three scales with Cronbach’s alpha. Reliability estimates, shown in \autoref{tab:scales}, support the use of the three scales, given the standard cutoff values for early-stage research (0.7) and basic research (0.8)~\cite{lance2006}.
As shown in \autoref{fig:ttiresults}, Welch’s $t$-tests revealed that although the mean ratings on the \textit{transparency}, \textit{trust} and \textit{intent to use} scales were slightly higher when GSC was used compared to Petal-X, these effects were not statistically significant. 
The Bayes factor for the same analysis revealed that the data were 4.29 (\textit{transparency}), 3.77 (\textit{intent to use}) and 2.35 (\textit{trust}) times more probable under the null hypotheses (no difference between the means of each group) compared to the alternative hypotheses. 
This can be considered, respectively, moderate evidence (in the case of \textit{transparency} and \textit{intent to use}) and anecdotical evidence (in the case of \textit{trust}) in favor of the null hypothesis\cite{jeffreys1961}.
\begin{figure}[h!t]
\centering
\includegraphics[width=\linewidth]{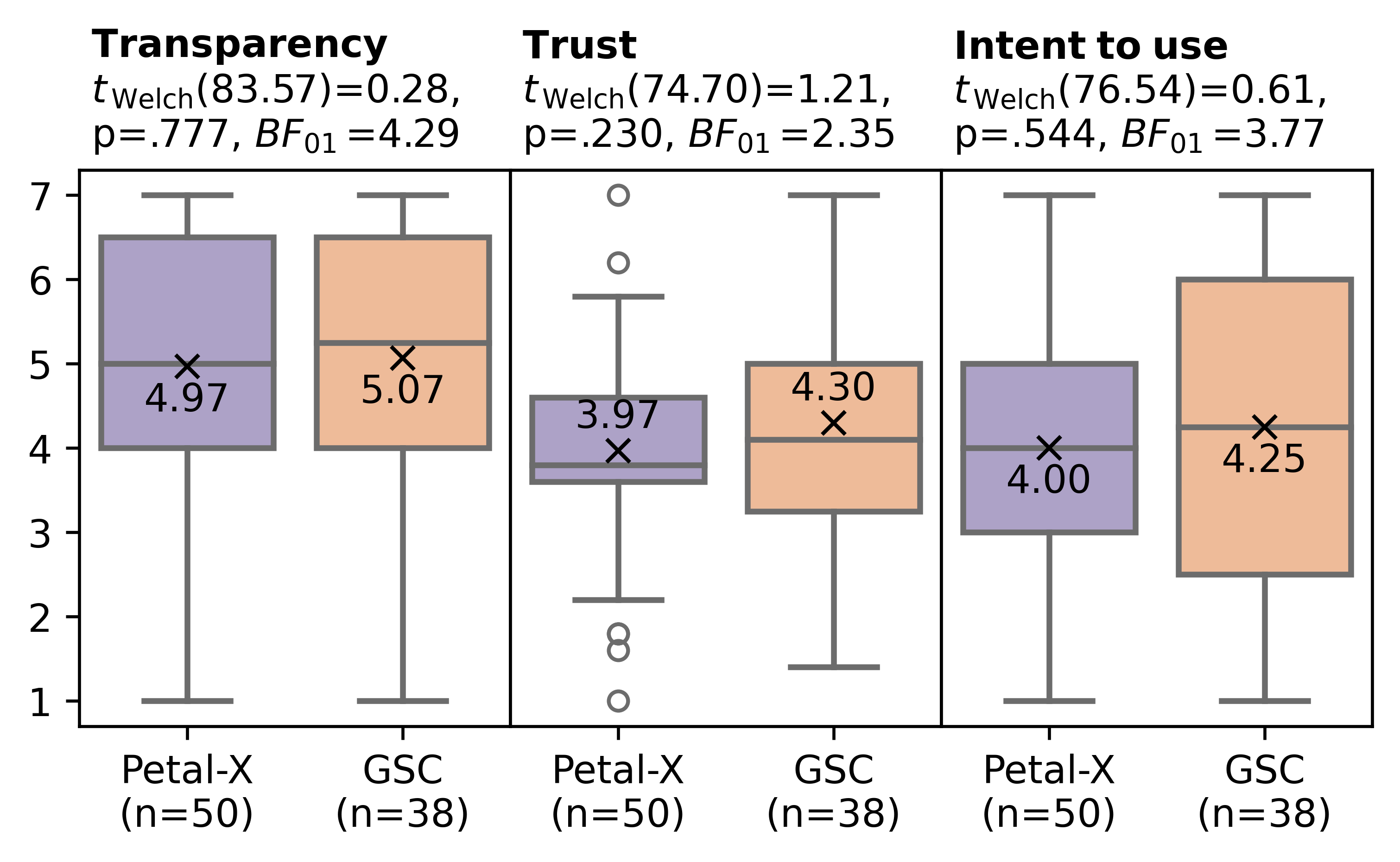}
\caption{Transparency, trust, and intent to use results.} 
\label{fig:ttiresults}
\end{figure}

\section{Discussion}\label{section:discussion} 
The results of this study reveal that Petal-X offers some performance benefits over GSCs in critical CVD risk communication challenges. Overall, healthcare students utilizing Petal-X (1) understand the contribution of each risk factor and (2) compare the projected risk reductions by risk factor modifications, both faster than those using GSCs and even lowering the error rate. These findings suggest that Petal-X's patient-centered design allows users to quickly grasp how the 10-year CVD estimate is derived from risk factors.

Petal-X would undoubtedly benefit from also numerically providing the total 10-year CVD risk and individual risk factor contributions. However, the area encoding provided by Petal Product Plots has been found outstanding for identifying risk factors of concern and may even be regarded as appropriate for computing the 10-year CVD risk, despite large variances in performance across users.

Finally, we found the evidence showing that there is little to no difference in transparency, trust, and intention to use between Petal-X and GSCs promising, especially given that these subjective impressions were formed after only brief use of Petal-X.
\section{Conclusion}
In this paper, we tackle the need for visual clinical risk explanations that target the different understanding capacities of non-ML expert audiences, such as clinicians, patients, and family members. We do so by designing and evaluating Petal-X, a tool focused on a specific but highly impactful issue: cardiovascular disease risk communication. The design process led to two intertwined contributions. First, we implemented a model-agnostic global surrogate tailored to the state-of-the-art SCORE2 model. Second, we designed Petal Product Plots, a novel visualization that mimics the shape of multilobe flower petals to propose a legible and engaging representation.

The use of model-agnostic XAI methodologies will potentially ease the adaptation of the proposed visualization to next-generation risk assessment models, which are expected to use more complex machine learning models to obtain a personalized risk prediction~\cite{tokgozoglu2021}. 

\acknowledgments{%
The authors wish to thank the participants for their time and contribution, as well as Jeroen Ooge and Alejandro Rozo for their helpful discussions. This work was supported by the Research Foundation-Flanders (FWO) grant G0A3319N and the Slovenian Research Agency grant ARRS-N2-010.
}

\bibliographystyle{abbrv-doi-hyperref}

\bibliography{main}

\end{document}